\documentclass[letterpaper, 10 pt, conference]{ieeeconf}

\IEEEoverridecommandlockouts                              

\usepackage{graphics} 
\usepackage{graphicx}
\usepackage{epsfig} 
\usepackage{comment}
\usepackage{subcaption}
\usepackage{cite}
\usepackage{amsmath} 
\usepackage{amssymb}  
\usepackage{array}

\usepackage{xcolor}
\usepackage[colorlinks=true]{hyperref}
\usepackage{makecell}

\let\labelindent\relax
\usepackage{enumitem}

\newtheorem{theorem}{Theorem}
\newtheorem{lemma}[theorem]{Lemma}

\newtheorem{xmpl}{Example}
\newtheorem{remark}{Remark}

\usepackage{algorithm}
\usepackage{algorithmic}

\title{\LARGE \bf Rollout-Based Approximate Dynamic Programming for MDPs with Information-Theoretic Constraints}

\author{Zixuan He$^{1}$, Charalambos D. Charalambous$^2$, and Photios A. Stavrou$^{1}$
\thanks{*The work of Z. He was supported by the Huawei France-EURECOM Chair on Future Wireless Networks. The work of P. A. Stavrou was supported in part by the SNS JU project 6G-GOALS \cite{strinati:2024} under the EU’s Horizon programme (Grant Agreement No. 101139232) and by the Huawei France-EURECOM Chair on Future Wireless Networks.}
\thanks{$^{1}$Z. He and P. A. Stavrou are with the Foundation and Algorithm Group, Communication Systems Department, EURECOM, France. {\tt\small email: \{zixuan.he,fotios.stavrou\}@eurecom.fr}}
\thanks{$^{2}$C. D. Charalambous is with the Department of Electrical and Computer Engineering, University of Cyprus, Cyprus.
{\tt\small email: chadcha@ucy.ac.cy}}%
}

\begin{document}

\maketitle
\thispagestyle{empty}
\pagestyle{empty}

\begin{abstract}
This paper studies a finite-horizon Markov decision problem with information-theoretic constraints, where the goal is to minimize directed information from the controlled source process to the control process, subject to stage-wise cost constraints, aiming for an optimal control policy. We propose a new way of approximating a solution for this problem, which is known to be formulated as an unconstrained MDP with a continuous information-state using Q-factors. To avoid the computational complexity of discretizing the continuous information-state space, we propose a truncated rollout-based backward-forward approximate dynamic programming (ADP) framework. Our approach consists of two phases: an offline base policy approximation over a shorter time horizon, followed by an online rollout lookahead minimization, both supported by provable convergence guarantees. We supplement our theoretical results with a numerical example where we demonstrate the cost improvement of the rollout method compared to a previously proposed policy approximation method, and the computational complexity observed in executing the offline and online phases for the two methods. 
	
\end{abstract}

\section{Introduction}
\label{section:introduction}

Sequential decision-making under uncertainty is a fundamental problem that attracts extensive attention across control and communication systems, paving the way for developing approximate dynamic programming (DP) or model-based reinforcement learning methods \cite{bertsekas:2019}. Markov decision processes (MDPs) provide a natural framework for modeling and analyzing sequential decision-making problems through the lens of approximate DP (ADP) problems \cite{puterman:2005}. 

Directed information (DI) is an information-theoretic measure that quantifies the causal information flow from one process to another in the presence of feedback information \cite{massey:1990}. It has been widely used by control and communication communities to quantify fundamental limits in feedback control systems under information-theoretic constraints \cite{charalambous:2017tac,tanaka:2018,kostina:2019a,stavrou:2021tac2}. By incorporating DI into the control objective, the controller is encouraged to be ``information-frugal", making decisions based on only the most relevant and necessary state information, thus enabling real-time systems to balance the performance tradeoffs between control and communication.

DI and other causal information-theoretic measures (eg. entropy, transfer-entropy \cite{schreiber:2000}) have been efficiently combined with MDPs to study sequential stochastic models using DP. For instance, \cite{savas:2019tac,savas:2022tac} considered the entropy as the reward in MDP and partially observable MDP (POMDP) to synthesize the control policies by maximizing the entropy. Other works integrate the information-theoretic quantities as regularized terms. 
In \cite{molloy:2023tac}, an entropy-regularized POMDP is used for active state estimation. Perhaps the most relevant work to ours is \cite{tanaka:2021tac}, which introduced the transfer-entropy-regularized MDP (TERMDP) for feedback control problems. Their approach seeks to minimize the control cost subject to a soft constraint, expressed through transfer entropy, a variant of DI, capturing the uncertainty between the state and control processes. Unfortunately, constructing the transfer-entropy regularizer with a standard state-dependent cost results in a nonconvex optimization problem, which makes the proposed dynamic forward–backward Blahut–Arimoto algorithm (BAA) \cite{blahut:1972} challenging to converge even to a local optimum, while incurring very high computational complexity.

Recently, we proposed a DP approach based on the information state, or belief, to study a lower bound for a zero-delay lossy compression problem by minimizing DI subject to a single-letter distortion function \cite{he:2024}. To address the well-known “curse of dimensionality” arising in the resulting DP recursion, we employed an ADP methodology grounded in {\it approximation in policy space} \cite{bertsekas:2019}. This ADP framework seeks suboptimal solutions by optimizing over a tractable, parameterized subset of the policy space. Although \cite{he:2024} did not explicitly involve control or MDPs, the DP techniques applied therein constitute a principled and widely adopted methodology for tackling information-theoretic MDPs and POMDPs, also known as information-state MDPs \cite{bertsekas:2019}.

Another major ADP approach, instead of directly targeting policies, is {\it approximation in value space}, which focuses on approximating the cost function. This strategy is particularly advantageous for MDPs with prohibitively large or continuous state spaces, as it enables more scalable methods that mitigate computational intractability (see, e.g., \cite{bertsekas:2019}). A notable class of methods within this paradigm is the rollout, which begins with an offline base-policy (heuristic) approximation and then applies online cost improvement. Rollout has demonstrated strong empirical effectiveness in reducing complexity while preserving performance guarantees \cite{bertsekas:2019}, and can also be interpreted as a form of model predictive control (MPC) that leverages a base policy \cite{bertsekas:2021rollout}.

{\it Contributions:} In this work, we consider a discrete-time finite-horizon directed information-constrained MDP (DI-MDP) problem in a feedback-controlled system (see Fig. \ref{Fig:system-MDP}), where the goal is to identify the optimal control policy that minimizes DI payoff subject to stage-wise state-dependent cost constraints. 
We first provide a new DP reformulation for the studied problem using a \textit{Q-factor} representation in terms of state and control policy pair, which can facilitate offline approximation for the base policy. 
We then propose a novel truncated-rollout-based backward-forward ADP framework that mitigates the computational burden of the studied problem by decoupling the offline and online phases. The offline training phase involves backward functional base policy and associated Q-factor approximation using a rolling or short horizon (see Algorithm \ref{Algo:Offline}). The online phase conducts forward rollout optimization with lookahead minimization followed by instantaneous Q-factor evaluations, to realize policy and cost improvement over any finite horizon (see Algorithm \ref{Algo:Online} and Theorem \ref{theorem:rollout-cost-improvement}). Both phases admit stage-wise convergence guarantees (see Lemmas \ref{Lemma:Convergence}, \ref{Lemma:stopping-criterion}). 
We corroborate our theoretical results with numerical simulations, demonstrating the cost improvement of the rollout method compared to a previously proposed policy approximation method \cite{he:2024} (see Fig. \ref{Fig:numerical-invariant-rollout}), and the computational complexity observed in executing the offline and online phases for the two methods (see Table \ref{Tab:computation-complexity}), where the guaranteed performance shows the scalability with increased computational resources.

\textit{Notation}: We denote the set of real numbers by $\mathbb{R}$. $\mathbb{N}\triangleq\{1,2,\ldots\}$, $\mathbb{N}_0\triangleq\{0,1,\ldots\}$, and $\mathbb{N}_j^N\triangleq\{j,\ldots,N\},~j\leq{N},~N\in\mathbb{N}$. We denote a sequence of random variables (RVs) by $X^t=\{X_0,X_1,\ldots,X_t\},t\in\mathbb{N}_0^N$ and their real values by $x^t\in\mathcal{X}^t=\{\mathcal{X}_0,\ldots,\mathcal{X}_t\}$, where $\mathcal{X}_t$ denotes the finite set and hence $\mathcal{X}^t$ the set sequence. A truncated sequence of RVs is denoted by $X_{j}^t=\{X_j,\ldots,X_t\}, j\in\mathbb{N}_0^t, {t}\geq{j}$, and its realizations by $x_{j}^t=\{x_j,\ldots,x_t\}\in\mathcal{X}_{j}^t=\{\mathcal{X}_{j},\ldots,\mathcal{X}_t\}$,~${t}\geq{j}$. The distribution of a RV $X$ on $\mathcal{X}$ is denoted by $P(x)$ and the conditional distribution of a RV $U$ given $X=x$ is denoted by $P(u|x)$. Functional dependencies are indicated with square brackets, e.g. $P[Q](x)$ and $P[u](x)$ represent a distribution $P(x)$ that depends on another distribution $Q$ or another realization $u$, respectively. The expectation operator is denoted by $\mathbb{E}\{\cdot\}$, and $\mathbb{E}^{\mu}\{\cdot\}$ denotes expectation with respect to a given distribution $\mu(\cdot)$.

\section{Problem Statement and Preliminaries}
\label{section:problem-statement}
We consider the controlled discrete-time dynamic system over any finite horizon, illustrated in Fig. \ref{Fig:system-MDP}. 
\begin{figure}[t]
    \centering
    \includegraphics[width=0.8\columnwidth]{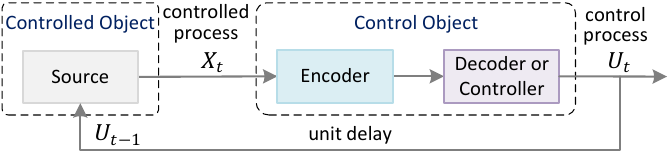}
    \caption{Controlled discrete-time dynamic system formulated as a finite-horizon MDP}
    \label{Fig:system-MDP}
\end{figure}
\subsection{System Model}
\label{system-model}
First, we describe each component of the system model in Fig. \ref{Fig:system-MDP}. Consider a controlled source process $X^t\in{\cal X}^t$ defined by a time index $t\in\mathbb{N}_0^N$ with given initial distribution $X_0\sim P_{0}(x_0)$ and a discrete control process $U^t\in{\cal U}^t$ in the system. We assume that the elements of the sets ${\cal X}^t$ and ${\cal U}^t$ are finite. At each time $t$, the controlled process $\{X_t:~t=0,\ldots, N\}$ is Markov conditional on the past controlled object $X_{t-1}$ and the past control ${U}_{t-1}=u_{t-1}$. The controlled process is defined by a transition kernel given by the conditional probability distribution
\begin{align}
    P_t(x_t|x_{t-1},u_{t-1}).
\label{Eq:Info-Trans-Kernel}
\end{align}
The control object or control process $\{U_t:t=0,\ldots, N\}$ is generated according to a randomized history-dependent control policy, hereinafter denoted by $\mu_t$, i.e., 
\begin{align}
    \mu_t = P_t(u_t|u^{t-1},x_t),
\label{Eq:Ctrl-Policy}
\end{align}
with the initial control policy $\mu_0$ to represent $P_0(u_0|x_0)$. Note that \eqref{Eq:Ctrl-Policy} depends only on the current controlled object $X_t$ since such a structural property was identified in \cite[Section IV]{charalambous:2014ecc}, \cite[Theorem 4.1]{stavrou:2018siam}.  On the other hand, for each $t=0,1,\ldots,N$, the controlled object 
$X_t$  is ``partially observed" through the control $U_t$ after applying the policy $\mu_t$ by the control object. The information of the controlled object filtered from the control process is summarized via an information state or belief given by the following recursion
\begin{align}
    &P_{t+1}(x_t|u^t) = \frac{P_t(x_t,u^t)}{P_{t+1}(u^t)}\nonumber\\
    &=\frac{\sum_{ x_{t-1}\in\mathcal{X}_{t-1}}\mu_t P_{t}(x_t|x_{t-1},u_{t-1})P_{t}(x_{t-1}|u^{t-1})}
    {\sum_{x_{t-1}^t\in\mathcal{X}_{t-1}^t}\mu_t P_{t}(x_t|x_{t-1},u_{t-1})P_{t}(x_{t-1}|u^{t-1})},
\label{Eq:Info-State-Update}
\end{align}
which is Markov, conditional on $u^{t-1}, P_{t}(x_{t-1}|u^{t-1})$. Finally, the control process $\{U_t:~t=0,1,\ldots, N\}$ induces the conditional probability distribution $P_t(u_t|u^{t-1}),~\forall{t}$, hereafter denoted by $\nu_t$, which is specified by the control object and information state via the following expression
\begin{align}
    \nu_t = P_t(u_t|u^{t-1})=&\sum_{x^t_{t-1}\in{\cal X}^t_{t-1}}\mu_t P_t(x_t|x_{t-1},u_{t-1})\nonumber\\
    &\qquad\times P_{t}(x_{t-1}|u^{t-1}).
\label{Eq:Output-Distribution}   
\end{align}

Here, by denoting the information state by $b_t$ and $P_t(x_t|x_{t-1},u_{t-1})$ by $w_t$, \eqref{Eq:Info-State-Update} can be expressed as a system transition equation of the form $b_{t+1}=f_t(b_{t},\mu_t,w_t)$.

\subsection{System Performance}
We specify the system performance using an information-theoretic cost as the objective function, subject to state-dependent cost constraints. Define the decision policy $\pi$ over the finite horizon $N$ as a sequence of control policies $\pi=\{\mu_0,\ldots,\mu_N\}$, and let the set of all admissible decision policies be $\Pi$ such that $\pi\in\Pi$. 
At each time stage $t\in\mathbb{N}_0^N$, the information-theoretic cost of the system incurred by the control policy $\mu_t$ is quantified by the conditional mutual information \cite{cover-thomas:2006}
\begin{align}
    I(X_t;U_t|&U^{t-1}) 
    \triangleq \mathbb{E}^{\mu_t} \Big\{\log\Big(\frac{\mu_t}{\nu_t}\Big)\Big\}\nonumber\\
    &=\sum_{\substack{x_{t-1}^t\in\mathcal{X}_{t-1}^t\\u^t\in\mathcal{U}^t}} \log\Big(\frac{\mu_t}{\nu_t}\Big) \mu_t w_t b_t P_t(u^{t-1}).
\label{Eq:Stage-Cost}
\end{align}
We define directed information \cite{massey:1990} theoretic cost as the additive cost functional 
\begin{align}
    C^{\pi}(X^N,U^N) \triangleq \sum_{t=0}^{N} I(X_t;U_t|U^{t-1}).
\label{Eq:Add-Cost}
\end{align}
Given a state-dependent cost function $\rho_t(x_t,u_t)$ at each $t$, defined over the state-control pair as $\rho_t:\mathcal{X}_t\times \mathcal{U}_t\to\mathbb{R}$, and a sequence of nonnegative constraints thresholds $\left\{D_0,\ldots,D_N\right\}$, the stage-wise state-dependent cost constraint associated with a control policy $\mu_t$ is specified by $\mathbb{E}^{\mu_t}\{\rho_t(X_t,U_t)\}\leq D_t$, where
\begin{align}
    \mathbb{E}^{\mu_t}\{\rho_t(X_t,U_t)\} =& \sum_{\substack{x_{t-1}^t\in\mathcal{X}_{t-1}^t\\u^t\in\mathcal{U}^t}}\rho_t(x_t,u_t)\mu_t w_tb_tP_t(u^{t-1}). 
\label{Eq:Info-Theo-Constraint}
\end{align}
Our DI-MDP is formulated as a constrained optimization problem, which seeks a decision policy $\pi$ that solves the following sequential stochastic optimization problem
\begin{align}
\begin{split}
    \inf_{\pi\in\Pi}& \quad C^{\pi}(X^N,U^N) \\
    \text{s.t.}& \quad \mathbb{E}^{\mu_t}\{\rho_t(X_t,U_t)\}\leq D_t,\forall t\in\mathbb{N}_0^N.
\end{split}
\label{Problem:Info-Theo-MDP}
\end{align}

We state some noteworthy remarks related to \eqref{Problem:Info-Theo-MDP}.

\begin{remark}
\label{remark:1}
{\bf (i)} Solving \eqref{Problem:Info-Theo-MDP} is equivalent to solving the lower bound (nonanticipative rate-distortion function) of the finite horizon sequential source coding problem with stage-wise distortion constraints \cite{stavrou:2018siam}, where the information-theoretic payoff \eqref{Eq:Add-Cost} interprets the expected compression rate, the control policy $\mu_t$ represents the test-channel (i.e., reconstruction probability distribution), and \eqref{Eq:Info-Theo-Constraint} as the stage-wise fidelity constraint.
{\bf (ii)} In contrast, the TERMDP problem in \cite{tanaka:2021tac} adopts a different formulation by reversing the objective and the constraint in \eqref{Problem:Info-Theo-MDP}; the payoff is regularized by transfer entropy. This leads to a nonconvex optimization problem (as discussed in \cite{tanaka:2021tac}), where the coupling between the information-regularization and the standard cost terms hinders the attainment of the global optimum.
\end{remark} 

\subsection{Reformulation of \eqref{Problem:Info-Theo-MDP}}

We now present the reformulation of problem \eqref{Problem:Info-Theo-MDP}.
We first leverage the key convexity properties of DI and $\Pi$ from \cite[Theorem 2, 3]{he:2024} with given $b_t$ obtained for fixed $U^{t-1}=u^{t-1}$ for any $t\in\mathbb{N}_0^N$. This allows to apply the Lagrange duality theorem \cite{boyd:2004} with Lagrange multipliers $\{s_t\leq0,~t\in\mathbb{N}_0^N\}$, converting \eqref{Problem:Info-Theo-MDP} into an unconstrained convex optimization problem. We then reformulate it as a finite-horizon \textit{information-state MDP} \cite{bertsekas:2019}, where for any $t\in\mathbb{N}_0^N$, each element of \textit{information state} $b_t$ takes values in the continuous interval $(0,1)$, resulting in a continuous information-state space encompassing an infinite number of possible distributions representing $b_t$.
The \textit{control policy} (action) $\mu_t$, the \textit{random disturbance} $w_t$, and the stage-wise \textit{cost function} $g_t(b_t,\mu_t)=\log\big( \frac{\mu_t}{\nu_t}\big)-s_t(\rho_t(x_t,u_t)-D_t)$ are defined accordingly. 
Table \ref{table:reformulated-variables} summarizes the correspondence between the components of this information-state MDP and those of the original problem \eqref{Problem:Info-Theo-MDP}.
\begin{table}
\centering
\caption{Information-State MDP Corresponds to \eqref{Problem:Info-Theo-MDP}}
\begin{tabular}{ll}
    \hline
    Information-state MDP&
    Connection to \eqref{Problem:Info-Theo-MDP}\\  \hline
    information state & $b_t=P_t(x_{t-1}|U^{t-1}=u^{t-1})$ \\
    control policy (action) & $\mu_t=P_t(u_t|U^{t-1}=u^{t-1},x_t)$ \\
    transition & $P_t(b_{t+1}|b_t,\mu_t)$ \\
    \makecell[l]{cost function\\(in contrast of reward)} & \makecell[l]{
        $g_t(b_t,\mu_t)$\\ $\quad~=\log\big( \frac{\mu_t}{\nu_t}\big)-s_t(\rho_t(x_t,u_t)-D_t)$
    } \\
    \hline
\label{table:reformulated-variables}
\end{tabular}
\end{table}

By invoking the \textit{principle of optimality} \cite{bertsekas:2019}, we obtain the following finite horizon stochastic DP recursion in Q-factor form, expressed explicitly in terms of any given information state $b_t$ and corresponding control policy $\mu_t$ pair, and computed backward in time $t=N,\ldots,0$
\begin{align}
    Q_t^*&(b_t,\mu_t)= \sum_{x_{t-1}^t\in\mathcal{X}_{t-1}^t,u_t\in\mathcal{U}_t}  \Big\{\Big(\log\Big( \frac{\mu_t}{\nu_t}\Big)- s_t \rho_t(x_t,u_t)\nonumber\\
    &+ \min_{\mu_{t+1}} Q_{t+1}^*(b_{t+1},\mu_{t+1})\Big) w_t\mu_t b_t \Big\}+s_t D_t,
\label{Eq:DP-Qfactor}
\end{align}
with the terminal condition $Q_{t+1}^*(b_{N+1},\mu_{N+1})=0$ for all $b_{N+1}$ and $\mu_{N+1}$. 
Note that the given information state $b_t$ for fixed $U^{t-1}=u^{t-1}$ appearing in \eqref{Eq:DP-Qfactor} can be determined at each time-stage $t$ through the recursion update \eqref{Eq:Info-State-Update}, and this update is Markovian conditioned on the current $U^{t-1}$ and $b_t$. Furthermore, at $t=0$, we adopt the initial assumptions of our system model by setting the initial control policy and marginal distribution as $\mu_0=P_0(u_0|x_0)$ and $P_0(u_0|u_{-1})=P_0(u_0)$ respectively. Hence, the recursion at the initial time stage in \eqref{Eq:DP-Qfactor} is adjusted accordingly.

The stochastic DP recursion given by \eqref{Eq:DP-Qfactor} can be solved using the backward-forward method from \cite{he:2024}, which relies heavily on finely discretizing the continuous information-state space into a large finite information-state space $\mathcal{B}_t$ per stage during the offline training phase. The offline computational complexity due to the curse of dimensionality of solving \eqref{Eq:DP-Qfactor}  motivates the usage of a novel, efficient ADP method introduced in the next section to obtain reachable results.

\section{Main Results}
\label{section:main-results}

In this section, we present our main results. We propose a truncated rollout-based ADP Framework that admits \textit{approximation in value space} to find the optimal solution.

The rollout approach is a sample-based heuristic method where an offline-approximated base policy is further improved in an online mechanism \cite{bertsekas:2019}. 
To obtain the offline-approximated base policy via sampling, we avoid computing $Q_t(b_t,\mu_t)$ backward over the entire horizon $N$. Instead, we operate over a truncated horizon, which naturally leads to a corresponding truncated rollout approach. We first introduce the approximate Q-factor $\tilde{Q}_t(b_t,\mu_t)$ at any $t\in\mathbb{N}_0^N$
\begin{align}
    &\tilde{Q}_t(b_t,\mu_t)
    =\tilde{Q}_t[g_t,\tilde{Q}_{t+1}](b_t,\mu_t)=\sum_{\substack{x_{t-1}^t\in\mathcal{X}_{t-1}^t\\u_t\in\mathcal{U}_t}}  \Big\{\Big(\log\Big( \frac{\mu_t}{\nu_t}\Big)
    \nonumber\\
    &- s_t \rho_t(x_t,u_t)+\tilde{Q}_{t+1}(b_{t+1},\mu_{t+1})\Big) w_t\mu_t b_t \Big\}+s_t D_t.
\label{Eq:approx-Qfactor}
\end{align}
Herein, we adopt a receding control approach by approximating the future $\tilde{Q}_{t+1}$ over a much shorter \textit{rolling horizon} $N_s$ (also as known as a sampling or \textit{receding horizon} \cite{bertsekas:2019}) compared with $N$ ($N_s\ll N$). Thus, $\tilde{Q}_{t}$ is rewritten as
\begin{align}
    &\tilde{Q}_t(b_t,\mu_t)= \tilde{Q}_t[g_t,\tilde{Q}_{N_s}^{\bar{\pi}}](b_t,\mu_t) =\sum_{\substack{x_{t-1}^t\in\mathcal{X}_{t-1}^t\\u_t\in\mathcal{U}_t}}  \Big\{\Big(\log\Big( \frac{\mu_t}{\nu_t}\Big)\nonumber\\
    &- s_t \rho_t(x_t,u_t)+{Q}_{N_s}^{\bar{\pi}}(b_{t+1},\mu_{t+1})\Big)w_t\mu_t b_t \Big\}+s_t D_t,
\label{Eq:Q-Factor-base-policy}
\end{align}
where ${Q}_{N_s}^{\bar{\pi}}$ is estimated backward over $N_s$, and $\bar{\pi}$ is the obtained base policy defined on a discretized information-state space denoted by $\bar{\mathcal{B}_t}$ such that $b_t\in\bar{\mathcal{B}_t}$ with its values selected from the interval $(0,1)$. 
Thus, the backward computation admits an offline training procedure with analytical recursions of the base control policy. To facilitate the backward approximation of $Q_{N_s}^{\bar{\pi}}$, we employ the following lemmas adapted from \cite{he:2024}.

\begin{lemma}{(Double minimization)}
\label{lemma:double-min}
For any $t\in\mathbb{N}_0^N$, let $s_t\leq0$ and $D_t>0$, then for a fixed $w_t$, and a given information state $\{b_t\in\bar{\mathcal{B}}_t\}$ obtained for a fixed $U^{t-1}=u^{t-1}$, minimizing \eqref{Eq:DP-Qfactor} admits a double minimum as
\begin{align}
    &Q_t^*(b_t,\mu_t^*) = \min_{\mu_t}\min_{\nu_t}\sum_{\substack{x_{t-1}^t\in\mathcal{X}_{t-1}^t\\u_t\in\mathcal{U}_t}}\Big\{\Big(\log\Big( \frac{\mu_t}{\nu_t}\Big)-s_t \rho_t(x_t,u_t) \nonumber\\
    &+Q_{t+1}^*(b_{t+1},\mu_{t+1})\Big)\mu_tw_t b_t \Big\}+ s_t D_t,~t=N,\ldots,0.
\label{Eq:DM-CostToGo}
\end{align}
Moreover, for a fixed $\mu_t$, the right hand side (RHS) of \eqref{Eq:DM-CostToGo} is minimized by \eqref{Eq:Output-Distribution}, whereas for fixed $\nu_t$, the RHS of (\ref{Eq:DM-CostToGo}) is minimized by 
\begin{align}
    \mu_t = \frac{\nu_tA_t[b_t](x_t,u_t,s_t)}{\sum_{u_t\in\mathcal{U}_t}\nu_tA_t[b_t](x_t,u_t,s_t)},
\label{Eq:AM-Policy}
\end{align}
where $A_t[b_t](x_t,u_t,s_t) = e^{s_t\rho_t(x_t,u_t)-Q_{t+1}^*(b_{t+1},\mu_{t+1})}$.
\end{lemma}

\begin{lemma}{(Base policy approximation algorithm)}
\label{Lemma:Convergence}
For each $t\in\mathbb{N}_0^N$, consider a fixed $w_t$, and a given $b_t$ obtained for a fixed $U^{t-1}=u^{t-1}$.
Moreover, for any $t$, let $s_t\leq0$ and $\nu_t^{(0)}$ be the initial output distribution with non-zero components, and let $\nu_t^{(k+1)}=\mu_t[\nu_t^{(k)}]$ and $\mu_t^{(k+1)}=\nu_t[\nu_t^{(k)}]$ be expressed wrt the previous $k^{\text{th}}$ iteration as follows
\begin{align}
    &\mu_t^{(k+1)}= \frac{\nu_t^{(k)}A_t[b_t](x_t,u_t,s_t)}{\sum_{u_t\in\mathcal{U}_t}\nu_t^{(k)}A_t[b_t](x_t,u_t,s_t)},
    \label{policy-update}
    \\
    &\nu_t^{(k+1)} = \nu_t^{(k)}\sum_{x_t\in\mathcal{X}_t}  \frac{b_tw_tA_t[b_t](x_t,u_t,s_t)}{\sum_{u_t\in\mathcal{U}_t}\nu_t^{(k)}A_t[b_t](x_t,u_t,s_t)}.
    \label{output-update}
\end{align}
Then as $k\to\infty$, we obtain for any $t$ that 
\begin{align*}
    Q_t(b_t,\mu_t^{(k)}) &\to Q_t^*(b_t,\mu_t^*),
\end{align*}
whereas $Q_t(b_t,\mu_t^{(k)})$ is expressed as
\begin{align}
    &Q_t(b_t,\mu_t^{(k)}) = \sum_{\substack{x_t\in\mathcal{X}_t,u_t\in\mathcal{U}_t}}b_tw_t\mu_t^{(k)} 
    \Big(\log\Big(\frac{\mu_t^{(k)}}{\nu_t^{(k)}}\Big)\nonumber\\&\qquad\qquad\qquad\qquad\qquad\qquad
    +Q_{t+1}^*(b_{t+1},\mu_{t+1})\Big).
\label{Eq:CostToGo-info.state-policy}
\end{align}
\end{lemma}

Therefore, to approximate $Q_{N_s}^{\bar{\pi}}$ over the rolling horizon $N_s$, the results provided by Lemma \ref{Lemma:Convergence} can be implemented via backward computation from the terminal stage $t=N$ to stage $t=N-N_s+1$, as detailed in Algorithm \ref{Algo:Offline}. The following lemma complements Algorithm \ref{Algo:Offline} by providing a stopping criterion to terminate the procedure after a finite number of iterations.
\begin{lemma}
\label{Lemma:stopping-criterion}
(Stopping criterion of Algorithm \ref{Algo:Offline}) 
For each $t\in\mathbb{N}_0^N$, the quantity $D_{s_t}$ given by 
\begin{align}
    D_{s_t}= \sum_{\substack{x_t\in\mathcal{X}_t u_t\in\mathcal{U}_t}}\mu_t^* w_tb_t\rho_t(x_t,u_t),
\label{distortion-point-s}
\end{align}
admits the following bounds
\begin{align}
    Q_t^*&(b_t,\mu_t)\leq  -\sum_{x_t\in\mathcal{X}_t}\big(b_tw_t\log(\sum_{u_t\in\mathcal{U}_t}\nu_tA_t[b_t](x_t,u_t,s_t))\big) \nonumber\\ 
    &-\sum_{u_t\in\mathcal{U}_t}\nu_tc_t[u^{t-1}](u_t)\log c_t[u^{t-1}](u_t) +s_t D_{s_t},
\label{Eq:CostToGo-upper-bound}\\
    Q_t^*&(b_t,\mu_t)\geq -\sum_{x_t\in\mathcal{X}_t}\big(b_tw_t\log(\sum_{u_t\in\mathcal{U}_t}\nu_tA_t[b_t](x_t,u_t,s_t))\big)\nonumber\\ 
    &-\max_{u_t\in\mathcal{U}_t}\log c_t[u^{t-1}](u_t)+s_t D_{s_t},
\label{Eq:CostToGo-lower-bound}
\end{align}
where $c_t[u^{t-1}](u_t)$ is expressed as a function of fixed $u^{t-1}$ 
\begin{align*}
    c_t[u^{t-1}](u_t) = &\sum_{x_t\in\mathcal{X}_t} \frac{b_tw_tA_t[b_t](x_t,u_t,s_t)}{\sum_{u_t\in\mathcal{U}_t}\nu_tA_t[b_t](x_t,u_t,s_t)}.
\end{align*}
\end{lemma}

Lemma \ref{Lemma:stopping-criterion} generates a stopping criterion for Algorithm \ref{Algo:Offline} at the $k$-th iteration by setting the estimation error $\epsilon$ per time stage, i.e., $T_{U_t}[u^{t-1},b_t]-T_{L_t}[u^{t-1},b_t]$, where
\begin{align*}
    T_{U_t}[u^{t-1},b_t] &= -\sum_{u_t\in\mathcal{U}_t}\nu_tc_t[u^{t-1}](u_t)\log c_t[u^{t-1}](u_t),\nonumber\\
    T_{L_t}[u^{t-1},b_t] &= -\max_{u_t\in\mathcal{U}_t}\log c_t[u^{t-1}](u_t).
\end{align*}

\begin{algorithm}[t]
\caption{Offline Base Control Policy Approximation}
\label{Algo:Offline}
\begin{algorithmic}[1]
    \REQUIRE given $\{w_t:t\in\mathbb{N}_{N-N_s+1}^N\}$,\\ given base information state $b_t\in\bar{\mathcal{B}}_t$, Lagrange multipliers $\{s_t\leq0:t\in\mathbb{N}_{N_s}^N\}$, error tolerance $\epsilon>0$
    \STATE {\textbf{Initialize} $\{\nu_t^{(0)}:t\in\mathbb{N}_{N-N_s+1}^N\}$\;}\\
    \FOR{$t=N:N-N_s+1$}
    \STATE $k\leftarrow0$
    \WHILE {$T_{U_t}[u^{t-1},b_t] - T_{L_t}[u^{t-1},b_t]  > \epsilon$}
    \STATE $\mu_{t}^{(k)}\leftarrow$ (\ref{policy-update}) \;
    \STATE $\nu_t^{(k+1)}\leftarrow$ (\ref{output-update})\;
    \STATE $Q_t(b_t,\mu_t^{(k)})\leftarrow$~\eqref{Eq:CostToGo-info.state-policy} 
    \STATE $k \leftarrow k + 1$\;
    \ENDWHILE
    \ENDFOR
    \STATE $Q_{N_s}^{\bar{\pi}}(b_t,\mu_t)\leftarrow Q^*_{N-N_s+1}[g_t,Q^*_{N-N_s+2}](b_t,\mu_t^*)$\;
    \ENSURE $\{\mu_{t}^*(b_t):t\in\mathbb{N}_{N-N_s+1}^N,b_t\in\bar{\mathcal{B}}_t\}$,\\
    $\{\nu_t^*[b_t]:t\in\mathbb{N}_{N-N_s+1}^N,b_t\in\bar{\mathcal{B}}_t\}$,\\
    $\{Q_{N_s}^{\bar{\pi}}(b_t,\mu_t):b_t\in\bar{\mathcal{B}}_t,\mu_t\in\mu_{t}^*(b_t)\}$. 
\end{algorithmic}
\end{algorithm}
\paragraph*{Comments on Algorithms \ref{Algo:Offline}} The control policy $\{\mu_{t}^*(b_t):t\in\mathbb{N}_{N-N_s+1}^N,b_t\in\bar{\mathcal{B}}_t\}$, the output distribution $\{\nu_t^*[b_t]:t\in\mathbb{N}_{N-N_s+1}^N,b_t\in\bar{\mathcal{B}}_t\}$, and the cost-to-go functions $\{Q_t^*:t\in\mathbb{N}_{N-N_s+1}^N\}$ are aprroximated over the base information-state space $\bar{\mathcal{B}}_t$ discretized under the fixed $U^{t-1}=u^{t-1}$. After computing these quantities backward over the rolling horizon $N_s$, we obtain $Q_{N_s}^{\bar{\pi}}(b_t,\mu_t)$ and proceed with online one-step truncated-rollout lookahead minimization starting from the initial information state $\tilde{b}_1 = P_0(x_0|u_0)$ and initial output distribution $\nu_0=P_0(u_0)$. The online procedure is detailed in Algorithm \ref{Algo:Online}.

\paragraph*{Comments on Algorithms \ref{Algo:Online}} In the online truncated-rollout approximation, for each $t$ with an associated information state ${b}_t$, the rollout control policy $\tilde{\mu}_t$ is determined by minimizing the approximate Q-factor $\tilde{Q}_t^{\bar{\pi}}({b}_t,\mu_t)$ averaged over $P_t(u^{t-1})$
\begin{align}
    \tilde{\mu}_t &= \arg\min_{\mu_t\in{\bar{\pi}}}\sum_{u^{t-1}\in\mathcal{U}^{t-1}} \tilde{Q}_t^{\bar{\pi}}({b}_t,\mu_t)P_t(u^{t-1}), 
\label{Eq:Online-Min}
\end{align}
where $\tilde{Q}_t^{\bar{\pi}}({b}_t,\mu_t)$ is obtained by one-step truncated-rollout lookahead minimizing approximate Q-factor \eqref{Eq:Q-Factor-base-policy}, expressed as
\begin{align}
    &\tilde{Q}_t^{\bar{\pi}}({b}_t,\mu_t) = \min_{\mu_t}\tilde{Q}_t[g_t,Q_{N_s}^{\bar{\pi}}]({b}_t,\mu_t).
\label{Eq:Approximate-Q-Factor}
\end{align}
Hence, the online rollout procedure at each $t$ involves first a functional minimization followed by an instantaneous policy evaluation. Specifically, the minimization in \eqref{Eq:Approximate-Q-Factor} corresponds to steps 3–9 of Algorithm \ref{Algo:Offline}, with $Q_t$ and $Q_{t+1}$ replaced by $\tilde{Q}_t$ and $Q_{N_s}^{\bar{\pi}}$ respectively, where convergence is guaranteed by Lemma \ref{Lemma:Convergence}. After determining the rollout policy $\tilde{\mu_t}$, the subsequent information state $\tilde{b}_{t+1}$ is updated via recursion in \eqref{Eq:Info-State-Update}. Consequently, our rollout approach yields an approximate solution to \eqref{Problem:Info-Theo-MDP}, generating an information-state trajectory $\{\tilde{b}_t,t\in\mathbb{N}_1^N\}$. The resulting rollout policy sequence, $\hat{\pi}=\{\mu_0,\tilde{\mu}_1,\ldots,\tilde{\mu}_N\}$, achieves the minimum defined by \eqref{Eq:Online-Min} for each $t\in\mathbb{N}_1^N$. We next show that the proposed rollout policy performs no worse than the base policy.
\begin{theorem}
\label{theorem:rollout-cost-improvement}
For any $t\in\mathbb{N}_1^N$, denote the by $\tilde{J}_t^{\hat{\pi}}(b_t)$ and $\tilde{J}_t^{\bar{\pi}}(b_t)$ the cost-to-go averaged over $P_t(u^{t-1})$ corresponding to the rollout policy $\hat{\pi}$ and base policy $\bar{\pi}$, respectively, starting at information state $b_t$, which are expressed as
\begin{align}
    &\tilde{J}_t^{\hat{\pi}}(b_t) = \sum_{u^{t-1}\in\mathcal{U}^{t-1}} \tilde{Q}_t^{\hat{\pi}}[g_t,Q_{N_s}^{\bar{\pi}}](b_t,\tilde{\mu}_t)P_t(u^{t-1}),
    \label{Eq:Ave-CTG-rollout}\\
    &\tilde{J}_t^{\bar{\pi}}(b_t) = \sum_{u^{t-1}\in\mathcal{U}^{t-1}} \tilde{Q}_t^{\bar{\pi}}[g_t,Q_{N_s}^{\bar{\pi}}](b_t,\mu_t)P_t(u^{t-1}).
\label{Eq:Ave-CTG-base}
\end{align}
Then $\tilde{J}_t^{\hat{\pi}}(b_t)\leq\tilde{J}_t^{\bar{\pi}}(b_t)$ for all $b_t$.
\end{theorem}
\begin{proof}
We prove this by induction backward in time. Clearly it holds for $t=N$ since $\tilde{J}_N^{\hat{\pi}}=\tilde{J}_N^{\bar{\pi}}$ according to \eqref{Eq:DP-Qfactor}. Assuming it holds for time stage $t+1$, we have for all $b_t$
\begin{align*}
    &\tilde{J}_t^{\hat{\pi}}(b_t)\\
    &\overset{(a)}{=} \sum_{u^{t-1}\in\mathcal{U}^{t-1}} \tilde{Q}_t^{\hat{\pi}}[g_t(b_t,\tilde{\mu}_t),Q_{N_s}^{{\bar{\pi}}}(\tilde{b}_{t+1})](b_t,\tilde{\mu}_t)P_t(u^{t-1})\\
    &\overset{(b)}{\leq} \sum_{u^{t-1}\in\mathcal{U}^{t-1}} \tilde{Q}_t^{\bar{\pi}}[g_t(b_t,\tilde{\mu}_t),Q_{N_s}^{\bar{\pi}}(\tilde{b}_{t+1})](b_t,\tilde{\mu}_t)P_t(u^{t-1})\\
    &\overset{(c)}{=} \min_{\mu_t\in\bar{\pi}}\\&\quad \sum_{u^{t-1}\in\mathcal{U}^{t-1}} \tilde{Q}_t^{\bar{\pi}}[g_t(b_t,{\mu}_t),Q_{N_s}^{\bar{\pi}}({b}_{t+1})](b_t,{\mu}_t)P_t(u^{t-1})
\end{align*}
\begin{align*}
    &\overset{(d)}{\leq} \sum_{u^{t-1}\in\mathcal{U}^{t-1}} \tilde{Q}_t^{\bar{\pi}}[g_t(b_t,{\mu}_t),Q_{N_s}^{\bar{\pi}}({b}_{t+1})](b_t,{\mu}_t)P_t(u^{t-1})\\
    &= \tilde{J}_t^{\bar{\pi}}(b_t),
\end{align*}
where $\overset{(a)}{=}$ is the DP equation for rollout policy \eqref{Eq:Ave-CTG-rollout}, $\overset{(b)}{\leq}$ holds because of the inductive assumption, $\overset{(c)}{=}$ holds by the definition of the online rollout minimization according to \eqref{Eq:Online-Min}, and $\overset{(d)}{\leq}$ holds by the DP equation for the base policy \eqref{Eq:Ave-CTG-base}. This completes the proof.
\end{proof}

Theorem \ref{theorem:rollout-cost-improvement} shows the cost improvement property and guarantees that the rollout policy does not degrade the performance compared with the base policy. Furthermore, repeated rollout can be performed by constructing an improved discretized information-state space $\tilde{\mathcal{B}}_t$ based on empirical observations from the previous rollout trajectory.
The updated base policy $\tilde{\pi}$ is re-approximated via Algorithm \ref{Algo:Offline} and subsequently enhanced via Algorithm \ref{Algo:Online}. Hence, repeated rollout can progressively refine the approximate optimal solution to problem (\ref{Problem:Info-Theo-MDP}).

\begin{algorithm}[t]
\caption{Online Rollout Evaluation}
\label{Algo:Online}
\begin{algorithmic}[1]
    \REQUIRE {$\{\bar{\mathcal{B}}_t:t\in\mathbb{N}_{N-N_s+1}^N\}$ of given $\{{b}_t:t\in\mathbb{N}_{N-N_s+1}^N\}$,\\
    $\{\mu_t^*(b_t):t\in\mathbb{N}_{N-N_s+1}^N, b_t\in\bar{\mathcal{B}}_t\}$, \\$\{\nu_t^*[b_t]:t\in\mathbb{N}_{N-N_s+1}^N,b_t\in\bar{\mathcal{B}}_t\}$,\\ $\{Q_{N_s}^{\bar{\pi}}:b_t\in\bar{\mathcal{B}}_t\}$.}
    \STATE \textbf{Initialize} $\mu_0=P_0(u_0|x_0)$, $P_1(u^0)$, $\tilde{b}_1=P(x_0|u_0)$ \;\\
    \FOR {$t = 1:N$}
    \STATE $\tilde{Q}_t^{\bar{\pi}}(\tilde{b}_t,\mu_t) \leftarrow$~step 3-9 in Algorithm \ref{Algo:Offline} \;
    \STATE $\tilde{\mu}_t \leftarrow$~\eqref{Eq:Online-Min} \;
    \STATE $\tilde{b}_{t+1}\leftarrow$~\eqref{Eq:Info-State-Update}\;
    \ENDFOR
    \ENSURE {$\hat{\pi}=\{\mu_0,\tilde{\mu}_1,\ldots,\tilde{\mu}_N\}$, $\{\tilde{b}_t,t\in\mathbb{N}_1^N\}$,\\ 
    $\{\tilde{\nu}_t:t\in\mathbb{N}_0^N\}$, $C^{\tilde{\pi}}(X^N,U^N)$.}
\end{algorithmic}
\end{algorithm}

\section{Numerical Example}
\label{numerical-results}
This section provides numerical simulations to support our theoretical findings that led to Algorithms \ref{Algo:Offline}, \ref{Algo:Online}. We consider binary alphabets, i.e., $\{\mathcal{X}_t=\mathcal{U}_t=\{0,1\}:t\in\mathbb{N}_0^N\}$, with Hamming distance function given by
\begin{equation}
    \rho_t(x_t,u_t)\equiv\rho(x_t,u_t) = \bigg\{
    \begin{matrix}
        0, \ \ \text{if}\ x_t= u_t \\
        1, \ \ \text{if}\ x_t\neq u_t
    \end{matrix}, ~~\forall{t}.
\end{equation}
We consider finite previous control memory (memory-1 $u^{t-1}\approx u_{t-1}$) applied for information state $b_t=P_t(x_{t-1}|u_{t-1})$, control policy $\mu_t=P_t(u_t|u_{t-1},x_t)$, and output distribution $\nu_t=P_t(u_t|u_{t-1})$. Under this assumption, we consider the base information-state space $\bar{\mathcal{B}}_t$, which consists of a matrix comprising two ``quantized'' binary probability distributions drawn from the original continuous state space. We denote with $n_t$ each quantization level per $t$, which results in an information-state space $\bar{\mathcal{B}}_t$ with size $|\bar{\mathcal{B}}_t| = n_t^{|\mathcal{U}_{t-1}|} = n_t^2$, representing combinations of $2$ out of $n_t$ quantized binary distributions.

\begin{xmpl}
\label{example}
(Time-invariant binary symmetric controlled Markov chain)~
The source distribution $w_t$ at each $t\in\mathbb{N}_1^N$ is chosen such that for each $t$, we have
\begin{equation}
\begin{split}
    w_t = 
    \begin{pmatrix}
        1-\alpha_0 & \alpha_0 & 1-\alpha_1 & \alpha_1\\
        \alpha_0 & 1-\alpha_0 & \alpha_1 & 1-\alpha_1
    \end{pmatrix},
\end{split}
\end{equation}
where the first two columns are under the condition of $u_{t-1}=0$ and the rest for $u_{t-1}=1$, where $\alpha_0,\alpha_1\in(0,1)$. Moreover, we choose the quantization levels $\{n_t=n:t\in\mathbb{N}_1^N\}$ and the stage-wise Lagrangian $\{s_t=s:t\in\mathbb{N}_0^N\}$. 
We demonstrate some results applying the proposed rollout with Algorithms \ref{Algo:Offline}, \ref{Algo:Online} by multi-core processing in Fig. \ref{Fig:info.cost} for $n=20$, $\alpha_0=0.4$, $\alpha_1=0.8$, $s=-2$, $N=100$, $N_S=5$.

\begin{figure}
\centering
\begin{subfigure}[b]{0.49\linewidth}
    \centering
    \includegraphics[height=3.4cm,width=\linewidth]{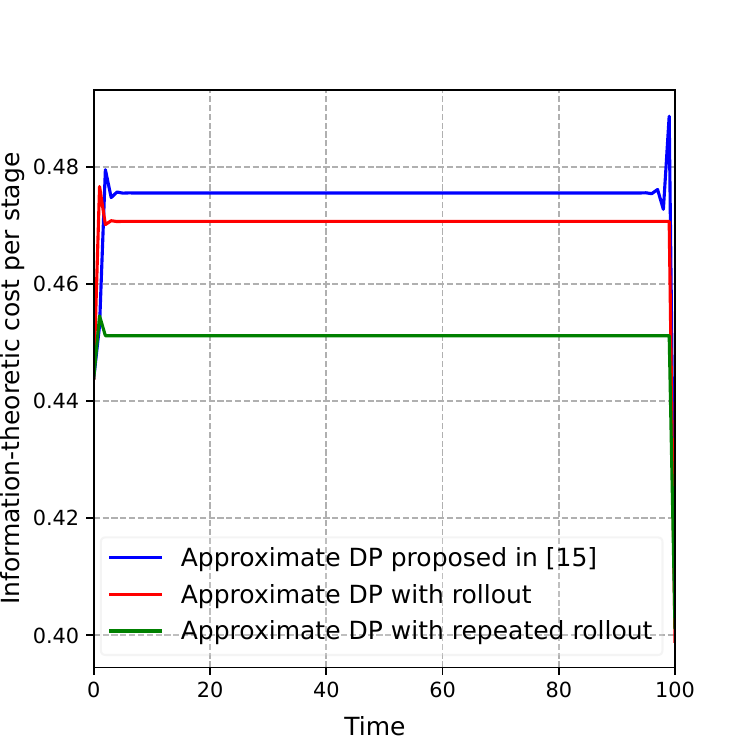}
    \caption{$I(X_t;U_{t}|U_{t-1})$}
    \label{Fig:info.cost}
\end{subfigure}
\begin{subfigure}[b]{0.49\linewidth}
    \centering
    \includegraphics[height=3.4cm,width=\linewidth]{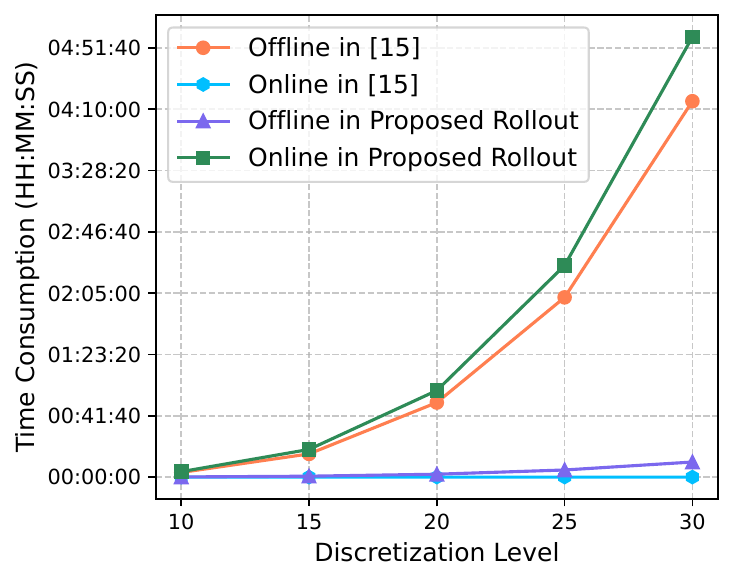}
    \caption{Time Consumption}
    \label{Fig:Time-Consumption}
\end{subfigure}
\caption{Illustration of the stage-wise information-theoretic cost and time consumption comparison}
\label{Fig:numerical-invariant-rollout}
\end{figure}

We then illustrate the computational complexity of our proposed scheme compared with the policy approximation approach in \cite{he:2024}. 
As shown in Table \ref{Tab:computation-complexity}, the proposed rollout-based method significantly reduces offline computation time by operating over a much shorter rolling horizon ($N_s\ll N$).  
Although the online rollout lookahead minimization of our proposed approach is a bit time-consuming, it offers a clear cost improvement over the prior approach in \cite{he:2024} under the same information-state space discretization level (See Fig. \ref{Fig:info.cost}, \ref{Fig:Time-Consumption}). 
We also demonstrate the cost performance of repeated rollout in Fig. \ref{Fig:info.cost} where the updated information-state space $\tilde{\mathcal{B}}_t$ with $|\tilde{\mathcal{B}}_t|=n_t^2$, is constructed based on the empirical value range. Using this refined $\tilde{\mathcal{B}}_t$ for a new round of offline base policy approximation, we observe a further cost improvement after the online rollout evaluation.

\begin{table}[t]
    \centering
    \begin{tabular}{ccc}\hline
     & Offline Phase & Online Phase \\\hline
    Proposed scheme & $\mathcal{O}(N_s n^2 \frac{1}{\epsilon})$ & $\mathcal{O}(N n \frac{1}{\epsilon})$ \\
    \cite{he:2024} & $\mathcal{O}(N n^2  \frac{1}{\epsilon})$ & $\mathcal{O}(Nn)$ \\\hline
    \end{tabular}
    \caption{Computation Complexity}
\label{Tab:computation-complexity}
\end{table}

We also observe that, with the comparable computational time, the proposed method requires indeed less random access memory (RAM) to realize approximation compared to the prior approach in \cite{he:2024}. This memory efficiency enables the offline base policy approximation to explore a larger discretized information-state space, thereby leading to improved online performance. Moreover, this efficiency enhances the proposed method's scalability, making it well-suited for extension to more complex settings as additional computational resources become available.
\end{xmpl}

\section{Conclusion}
\label{conclusion}
We proposed a DI-MDP problem over a finite horizon, aiming to derive the optimal control policy by minimizing the DI subject to stage-wise state-dependent cost constraints. We first reformulated the problem as an unconstrained information-state MDP and derived DP recursion in Q-factor form. To address the intractability caused by the continuous information-state space in the offline training phase, we proposed a new rollout-based ADP framework that involves offline base policy approximation followed by online approximate policy improvement. Our theoretical results are supplemented with simulation studies that showed the cost performance improvement and demonstrated the scalability.

\bibliographystyle{IEEEtran}
\bibliography{string,literature_conf}





\end{document}